\documentclass[superscriptaddress,twocolumn,amsmath,amssymb,pra]{revtex4-1}
\usepackage{float}
\usepackage{graphicx}
\usepackage{dcolumn}
\usepackage{bm}
\usepackage{amsmath}
\usepackage[colorlinks=true,citecolor=blue,urlcolor=black]{hyperref}
\usepackage{booktabs}
\usepackage{adjustbox}
\usepackage{array}

\newtheorem{theorem}{Theorem}
\newtheorem{lemma}{Lemma}
\newtheorem{proof}{Proof}
\usepackage{xcolor}

\def\ket#1{| #1 \rangle}
\def\bra#1{\langle #1 |}

\def\R{\mathbb{R}}
\def\O{\mathcal{O}}
\def\poly{\operatorname{poly}}

\definecolor{blue}{rgb}{0,0,1}
\definecolor{red}{rgb}{1,0,0}
\definecolor{green}{rgb}{0,1,0}

\begin{document}


\title{Fixed-point quantum continuous search algorithm with optimal query complexity}

\author{Shan Jin}
\affiliation{Institute of  Fundamental and Frontier Sciences, University of Electronic Science and Technology of China, Chengdu, Sichuan, 610051, China}

\author{Yuhan Huang}
\affiliation{Institute of  Fundamental and Frontier Sciences, University of Electronic Science and Technology of China, Chengdu, Sichuan, 610051, China}
\affiliation{Department of Electronic and Computer Engineering, The Hong Kong University of Science and Technology, 999077, Hong Kong}

\author{Shaojun Wu}
\affiliation{Institute of  Fundamental and Frontier Sciences, University of Electronic Science and Technology of China, Chengdu, Sichuan, 610051, China}

\author{Guanyu Zhou}
\email{zhoug@uestc.edu.cn}
\affiliation{Institute of  Fundamental and Frontier Sciences, University of Electronic Science and Technology of China, Chengdu, Sichuan, 610051, China}

\author{Chang-Ling Zou}
\email{clzou321@ustc.edu.cn}
\affiliation{Key Laboratory of Quantum Information, CAS, University of Science and Technology of China, Hefei, China}
\affiliation{Hefei National Laboratory, Hefei 230088, China}

\author{Luyan Sun}
\email{luyansun@tsinghua.edu.cn}
\affiliation{Hefei National Laboratory, Hefei 230088, China}
\affiliation{Center for Quantum Information, Institute for Interdisciplinary Information Sciences, Tsinghua University, Beijing 100084, China}

\author{Xiaoting Wang}
\email{xiaoting@uestc.edu.cn}
\affiliation{Institute of  Fundamental and Frontier Sciences, University of Electronic Science and Technology of China, Chengdu, Sichuan, 610051, China}

\date{\today}

\begin{abstract}
Continuous search problems (CSPs), which involve finding solutions within a continuous domain, frequently arise in fields such as optimization, physics, and engineering. Unlike discrete search problems, CSPs require navigating an uncountably infinite space, presenting unique computational challenges. In this work, we propose a fixed-point quantum search algorithm that leverages continuous variables to address these challenges, achieving a quadratic speedup. Inspired by the discrete search results, we manage to establish a lower bound on the query complexity of arbitrary quantum search for CSPs, demonstrating the optimality of our approach. In addition, we demonstrate how to design the internal structure of the quantum search oracle for specific problems. Furthermore, we develop a general framework to apply this algorithm to a range of problem types, including optimization and eigenvalue problems involving continuous variables.
\end{abstract}

\maketitle


\section{Introduction}

Continuous search problems (CSPs) focus on finding solutions within continuous domains, distinguishing them from discrete search problems that involve finite or countable sets. CSPs are common in fields where variables take values from uncountably infinite spaces, such as constrained optimization~\cite{constrained_opti}, nonconvex optimization~\cite{nonconvex-opti}, derivative-free optimization~\cite{deriv_free_opti}, energy landscape~\cite{energy_lands}, as well as continuous-variable (CV) spectral problem~\cite{quant_math_hall}. The challenges of CSPs arise from the need to explore vast, often high-dimensional continuous spaces efficiently, which introduces huge computational complexities not present in discrete problems.

With the advent of quantum computing, CSPs have gained new momentum due to the potential for quantum speedup. Quantum algorithms, such as the Deutsch-Jozsa algorithm~\cite{deutsch1985quantum, deutsch1992rapid}, Simon's algorithm~\cite{simon1997power}, Shor's algorithm~\cite{shor1994algorithms}, Grover's search algorithm~\cite{Grover:1996:FQM:237814.237866}, and the HHL algorithm \cite{harrow2009quantum}, have demonstrated the advantages of quantum computation over classical computation. In particular, Grover's quantum search algorithm offers a quadratic speedup compared to classical algorithms for solving a wide range of problems~\cite{Grover:1996:FQM:237814.237866,bennett1997strengths,brassard1997exact, brassard1998quantum,grover1998quantum,boyer1998tight,zalka1999grover,long1999phase,long2001grover,long2002phase,brassard2002quantum,grover2005fixed,yoder2014fixed, Liu2024, He2024}. Furthermore, the concept of the quantum search algorithm can be extended to the amplitude amplification algorithm, facilitating quantum state preparation~\cite{brassard1997exact, grover1998quantum,brassard2002quantum}. In spite of great success, due to the lack of precise determination of the number of marked states, the original Grover's quantum search cannot effectively avoid ``undercooking" or ``overcooking" the initial states \cite{brassard1997searching}. To address these issues, fixed-point quantum search algorithms have been proposed, which allow the quantum state to converge towards the target state without prior knowledge of the number of target items \cite{grover2005fixed, yoder2014fixed}. Several experimental platforms have been proposed for implementing quantum search algorithms, including cavity quantum electrodynamics (QED) \cite{yamaguchi2002quantum, deng2005simple, yang2007implementation}, trapped ions \cite{feng2001grover, brickman2005implementation, ivanov2008simple}, nuclear magnetic resonance (NMR) \cite{chuang1998experimental, vandersypen2000implementation}, and optics \cite{bhattacharya2002implementation}.

Most established quantum search algorithms primarily focus on discrete search problems. Directly applying discrete quantum search to solve CSPs poses challenges due to the necessity of discretization. CSPs involve searching for solutions within continuous domains in $\R^d$. The discretization process requires $(\frac{1}{\epsilon})^d$ points, where $\epsilon$ represents the precision of the error and depends on the specific search problem. Complex problems demand higher accuracy level, i.e., smaller values for $\epsilon$, thereby increasing the complexity. In contrast to discrete-variable (DV), e.g. qubit-based quantum computing models, CV quantum computing models~\cite{lloyd1999quantum} possess unique characteristics and they are ideal models to deal with CV-related problems~\cite{RevModPhys.84.621,RevModPhys.77.513,CVQI,PhysRevA.65.042304,PhysRevLett.101.130501,PhysRevA.83.052325,Su2013,PhysRevLett.97.110501,PhysRevA.76.032321,Loock:07,PhysRevA.79.062318,PhysRevLett.107.250501,PhysRevLett.98.070502,Liu2020,PhysRevLett.132.100801,PhysRevLett.132.040601}. In fact, a CV-based quantum search algorithm has been proposed to solve discrete search problems~\cite{pati2000quantum}, which is unfortunately not immediately applicable to CSPs.

To address this problem, in this work, we propose a CV fixed-point quantum search algorithm to solve CSPs while maintaining a quadratic speedup. Its fixed-point property ensures the convergence to the target solution without the issues of ``undercooking" or ``overcooking". Similar to the quadratic speedup proven to be optimal in discrete quantum search ~\cite{bennett1997strengths, boyer1998tight,zalka1999grover, Grover:2005:PQS:1073970.1073997,dohotaru2008exact}, we show that our algorithm achieves an optimal quadratic speedup for CSPs. Specifically, the complexity of our algorithm is $\mathcal{O}(\frac{1}{\sqrt{\lambda}})$, where $\lambda$ is the overlap between the target solution space and the search space. We prove that the lower bound on query complexity of an arbitrary query-based quantum algorithm to solving CSPs is $\frac{1}{2\sqrt2}[(1+\sqrt p -\sqrt {1-p})\sqrt{n}-2]$ queries, where $n\equiv \lceil \frac{1}{\lambda} \rceil$ and $p$ is the probability of success. This confirms the optimality of our algorithm for CSPs. We also present how to design the internal structure of oracles for the desired CSPs and provide a practical scheme for their physical construction. Finally, we demonstrate applications of our algorithm to optimization and eigenvalue problems involving CVs, illustrating its versatility and potential impact across various domains.

\begin{figure*}
\centering
\includegraphics[width=5.3in]{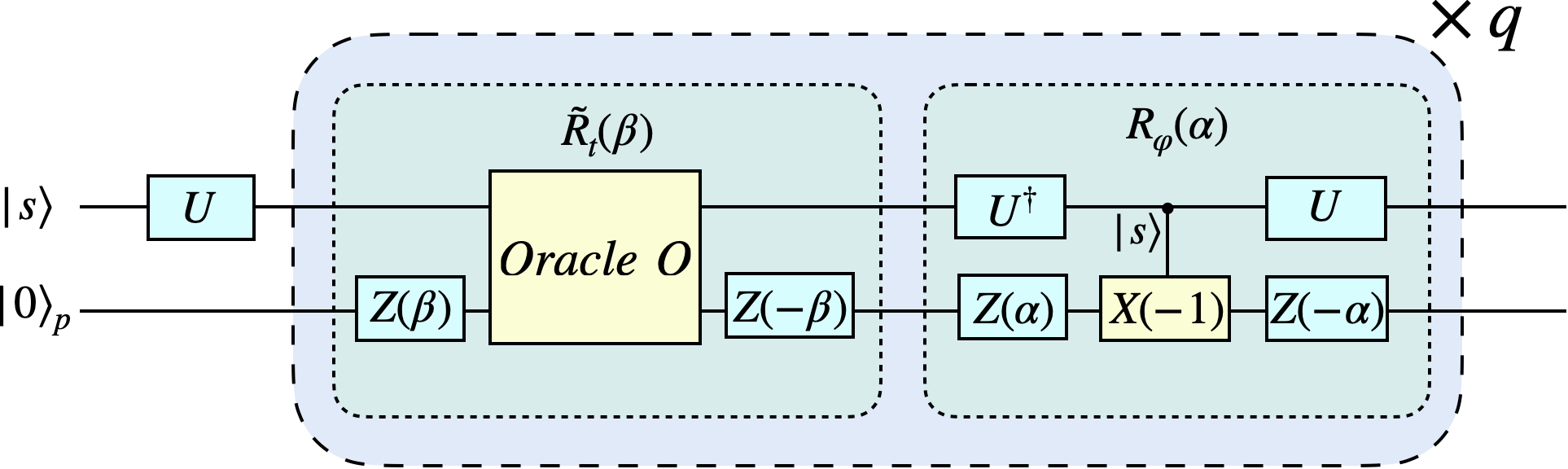}
\caption{The quantum circuit of the proposed quantum search algorithm for CSPs. The circuit demonstrates a sequence of $q$ Grover's iterations with each iteration $G(\alpha,\beta) \equiv R_\varphi(\alpha) R_t(\beta)$, and $Z(\alpha)\equiv e^{-i\alpha\hat x}$ and $X(\beta)\equiv e^{-i\beta\hat p}$ are defined as the parametrized translation operators generated by $\hat x$ and $\hat p$. The left part of the circuit implements $\tilde{R}_t(\beta)$, which is equal to $R_t(\beta)$ when restricted to the 2D subspace $\mathcal{T}$. The right part implements $R_{\varphi}(\alpha)$, where $\ket{\varphi}$ is the initial state with $\ket{\varphi}=U\ket{s}$. The unitary operation with $\ket{s}$ as the control is $|s\rangle\langle s|\otimes X(-1) + (I - |s\rangle\langle s|)\otimes I$, with $X(-1)=e^{i\hat{p}}$.
}\label{fig-circ-search}
\end{figure*}

\section{Fixed-point quantum continuous search design}\label{sec:search_design}

In a CSP, we assume the search space is a compact, measurable region $A\subset\R^d$. In addition, we assume that the solution to the CSP constitutes a measurable sub-region $Q\subset A$, and is unknown to us prior to search. Our objective is to identify and locate at least one point in $Q$. We assume the availability of an oracle $f$ which indicates whether a given input $x\in A$ belongs to $Q$ or not. Specifically, $f$ is an indicator function, $f: \bm{x}\in A \to \{0,1\}$ where $f(\bm x)=1$ if $\bm{x} \in Q$ and $f(\bm x)=0$ if $\bm{x} \in A/Q$. Letting $\lambda_0 \equiv \frac{m(Q)}{m(A)}$, where $m(\cdot)$ denotes the measure, the average query complexity of finding out a solution in an unstructured CSP through oracle evaluation is $\O(\frac{1}{\lambda_0})$. Our goal is to design a quantum continuous search algorithm that achieves a quadratic speedup, reducing the query complexity to $\O(\frac{1}{\sqrt{\lambda_0}})$. Specifically, in Sec.~\ref{sec_algorithm_design}, based on the unique features of CV quantum computing, we design the fundamental structure of the fixed-point quantum continuous search algorithm, including the choices of the initial and the target states, as well as the construction of the Grover's iteration, composed of two consecutive rotation gates restricted in a 2-dimensional subspace. In Sec.~\ref{sec-impl}, we focus on the specific design of the quantum continuous search algorithm for a given CSP, including the construction of the quantum oracle and the design of the associated quantum circuits.

\subsection{Algorithm Design}\label{sec_algorithm_design}

In this work, we design the circuit of the fixed-point quantum continuous search algorithm to solve CSPs, as illustrated in Fig.~\ref{fig-circ-search}. Let's start with the initial-state preparation. Recall that for a discrete search problem, the search space $A$ consists of $N$ elements, each encoded as one of the computational basis states $\{\ket{k}\}$, $k=0,\cdots,N-1$. Suppose there are $M$ solutions to the search problem, forming a subset $Q\subset A$. The target state is defined as $\ket{t}\equiv \frac{1}{\sqrt{M}}\sum_{x\in Q} \ket{x}$. In a quantum search algorithm, the process begins by preparing the quantum register into an initial state $\ket{\varphi_0}$, followed by applying a sequence of Grover's iterations to evolve the initial state into the final state $\ket{\varphi_{f}}$. The objective is to maximize the overlap between the final state $\ket{\varphi_{f}}$ and target state $\ket{t}$. However, since the target state $\ket{t}$ is unknown prior to the search, the optimal strategy is to initialize the state as an equal superposition of all potential solution candidates: $\ket{\varphi_0}\equiv \frac{1}{\sqrt{N}}\sum_{k=0}^{N-1}\ket{k}$. This ensures that all potential solutions are represented equally, providing a uniform starting point for the search process. Likewise, for quantum continuous search problems, the goal is to encode all points within the set $A$ into an equal or approximately equal superposition state:
\begin{align}\label{eq-ini-s}
\ket\varphi\equiv \int_{\R^d} \varphi (\bm x) \ket{\bm{x}} d\bm x= \int_A \varphi (\bm x) \ket{\bm{x}} d\bm x,
\end{align}
where $\varphi (\bm{x})$ is square-integrable and $\varphi (\bm{x})=0$ on $\R^d/A$. The choice of $\varphi (\bm{x})$ is flexible, with the only requirement that in $A$ there is no sub-region for $\varphi (\bm x)$ to have extremely low probability density. For a given CSP, we can choose the initial state to be an equal superposition state $\ket\varphi=\ket{\varphi_e}$ with $\varphi_e (\bm{x})= \frac{1}{\sqrt{m(A)}}$ on $A$ and zero elsewhere. In practice, the CV system can be realized with harmonic oscillators, and we can construct a superposition of Fock states $\ket{\varphi_s}=\sum_{k=1} a_k\ket{k}$ to closely approximate $\ket{\varphi_e}$, and the initial state $\ket\varphi$ can be efficiently prepared through an engineered universal unitary gate $U$ on an accessible source state $\ket{s}$ as $\ket\varphi=U\ket{s}$~\cite{Wang2019}, where $\ket{s}$ can be the vacuum state.

In the following, $Q$ is referred to as the search target space, gives the target state
\begin{align}
\left| t \right\rangle \equiv \int_{A} \varphi_1 (\bm{x}) \ket{\bm x}  d\bm{x}=\frac{1}{\sqrt{I_1}}\int_{Q} \varphi (\bm{x}) \ket{\bm x}  d\bm{x},
\end{align}
where $I_1=\int_{Q} |\varphi (\bm{x})|^2 d\bm{x}$ and $\varphi_1 (\bm{x})=\varphi (\bm{x})/\sqrt{I_1}$ on $Q$ and zero on $A/Q$. Intuitively, $\ket t$ is derived from $\ket\varphi$ and represents a superposition of all solution states supported within $A$. Although it remains unknown prior to the search, $\ket t$ is a crucial concept in constructing the Grover's iteration circuit. We also define the overlap between the two states:
\begin{align}
\lambda\equiv |\bra t \varphi\rangle|^2=\int_{Q} \left| \varphi(\bm{x}) \right|^2 d\bm{x}.
\end{align}
In CSPs, $\lambda$ represents the area ratio of the target region $Q$ to the entire search area $A$, which is an important parameter to quantify the query complexity of our quantum search algorithm. When $\ket{\varphi}$ is chosen to be the equal superposition state $\ket{\varphi_e}$, we have $\lambda=\lambda_0=\frac{m(Q)}{m(A)}$.

We also introduce the state $\ket{\bar t}$ orthogonal to $\ket{t}$:
\begin{align}
\ket{\bar t}   \equiv \int_{A} \varphi_{2}(\bm{x}) \left| \bm{x} \right\rangle d\bm{x}=\frac{1}{\sqrt{I_2}}\int_{A/Q} \varphi (\bm{x}) \ket{\bm x}  d\bm{x},
\end{align}
where $I_{2}=\int_{A/Q} |\varphi (\bm{x})|^2 d\bm{x}$ and $\varphi_{2} (\bm{x})=\varphi (\bm{x})/\sqrt{I_{2}}$ on $A/Q$ and zero on $Q$. Then corresponding overlap with  $\ket{\varphi}$ is $\bar{\lambda} \equiv |\bra{\bar t} \varphi\rangle|^2= \int_{A/Q} \left| \varphi(\bm{x}) \right|^2 d\bm{x}$. $\ket{t}$ and $\ket{\bar{t}}$ generate a two-dimensional subspace $\mathcal{T}$ in the original Hilbert space. We can express the initial state $\ket{\varphi}$ in terms of $\ket{t}$ and $\ket{\bar{t}}$:
\begin{align}\label{eq1}
\begin{split}
\ket{\varphi} = \sqrt{\bar{\lambda}} \ket{\bar{t}} + \sqrt{\lambda} \ket{t} \equiv \begin{pmatrix}
\sqrt{\bar{\lambda}} \\ \sqrt{\lambda}
\end{pmatrix}_\mathcal{T}.
\end{split}
\end{align}
Next, given the CSP and the classical oracle $f$, we can construct the corresponding quantum oracle $O$:
\begin{align} \label{eq-oracle}
O\ket{\bm{x}}\ket{y}=\ket{\bm{x}}\ket{y+f(\bm{x})}=
\begin{cases}
\ket{\bm{x}}\ket{y+1}, \ & \bm{x}\in Q \\
\ket{\bm{x}}\ket{y}, & \bm{x}\notin Q
\end{cases}.
\end{align}
Here, $\ket{y}$ represents the eigenstate of the position operator $\hat{x}$, with $\hat{x}\ket{y}=y\ket{y}$. The role of the quantum oracle $O$ is to determine whether the input $x$ belongs to $Q$ or not: if it does, the state of the auxiliary mode $\ket{y}$ is displaced by $1$; otherwise, it remains unchanged.

As illustrated in Fig.~\ref{fig-circ-search}, the quantum search circuit is constructed as a sequence of $q$ Grover's iterations:
\begin{align}
U^{(q)} = G(\alpha_q, \beta_q) \cdots G(\alpha_1, \beta_1).
\end{align}
Each iteration $G(\alpha,\beta)\equiv R_\varphi(\alpha) R_t(\beta)$ is composed by two parameterized inversion operators with respect to $\ket{\varphi}$ and $\ket{t}$, defined as:
\begin{subequations}
\begin{align}
R_\varphi(\alpha) &\equiv I -(1-e^{i\alpha})\ket{\varphi}\bra{\varphi},\\
R_t(\beta) &\equiv I-(1-e^{i\beta})\ket{t}\bra{t}.
\end{align}
\end{subequations}
Since $\ket{t}$ and $\ket{\varphi}$ both belong to the 2D subspace $\mathcal{T}$, $R_{\varphi}(\alpha)$ and $R_t(\beta)$ are reduced to $2$-by-$2$ operators on $\mathcal{T}$:
\begin{subequations}
\begin{align}
\begin{split}\label{eq2}
R_{\varphi}\left( \alpha \right) &= I - \left(1 - e^{i\alpha}\right) \ket{\varphi} \bra{\varphi} \\
&= I - (1 - e^{i\alpha}) \begin{pmatrix}
\bar{\lambda} & \sqrt{\lambda \bar{\lambda}} \\
\sqrt{\lambda \bar{\lambda}} & \lambda
\end{pmatrix}_\mathcal{T},
\end{split} \\
\begin{split}\label{eq3}
R_{t}\left( \beta \right) &= I - \left(1 - e^{i\beta}\right) \ket{t} \bra{t}
\equiv \begin{pmatrix}
1 & 0 \\
0 & e^{i\beta}
\end{pmatrix}_\mathcal{T}.
\end{split}
\end{align}
\end{subequations}
It is evident that all intermediate states $\ket{\varphi^{(l)}} = U^{(l)} \ket{\varphi}$ for $l=1,\ldots,q$, belong to $\mathcal{T}$.

There are various ways of designing the Grover's iteration $G(\alpha_j, \beta_j)$. A fixed-point quantum search satisfies that, by choosing appropriate values of $\alpha_j$ and $\beta_j$, $j=1,2,\cdots,q$, $U^{(q)}$ will drive the system state to converge towards the target state $\ket{t}$, where $q$ is the number of queries. When $\alpha_j=\beta_j=\pi$, $U^{(q)}$ reduces to the original non-fixed-point Grover search algorithm. In order to prevent the ``overcooking" or the ``undercooking" problems while still achieving the quadratic speedup, we adopt the approach in \cite{yoder2014fixed} to choose appropriate values of $\{\alpha_j, \beta_j\}$. The angle parameters in the quantum search algorithm need to satisfy the condition $\alpha_j = \beta_{q-j+1} = -2\cot^{-1}(\tan(2\pi j / L)\sqrt{1 - \eta^2})$, where $\eta^{-1} = T_{1/L}(1/\sqrt{\delta})$, $L = 2q + 1$ with $q$ denoting the number of query iterations. Here, $T_L(x) = \cos[L\cos^{-1}(x)]$ represents the $L^{\text{th}}$ Chebyshev polynomial of the first kind and $\delta$ corresponds to the desired accuracy for obtaining search results. Based on Eqs.~\eqref{eq1}, \eqref{eq2}, and \eqref{eq3}, it can be concluded that implementing a CV-based quantum search algorithm is feasible within subspace $\mathcal{T}$. Given an error threshold $\delta$, from $|\langle t|U^{(q)}|\varphi\rangle|^2>1- \delta$, we obtain the number of queries required for achieving the target state within $\delta$:
\begin{align}\label{eq4}
q \ge \frac{1}{2}\Big(\frac{\log(2/\sqrt{\delta})}{\sqrt{\lambda}} - 1\Big).
\end{align}
In other words, the query complexity of our fixed-point quantum search algorithm is $\mathcal{O}(\frac{1}{\sqrt{\lambda}})$, which offers a quadratic quantum speedup compared to classical query-based algorithms in solving CSPs.

We emphasize that the choice of parameters $(\alpha_j,\beta_j)$ is crucially important in the algorithm design to gain quantum speedup. In Appendix \ref{sec-pi3-search}, we have included an adaptation of Grover's original $\pi/3$ fixed-point search algorithm to solve CSPs, whose query complexity is only $\mathcal{O}(1/\lambda)$, with no quantum advantage.

\begin{figure}
\centering
\includegraphics[width=3.4in]{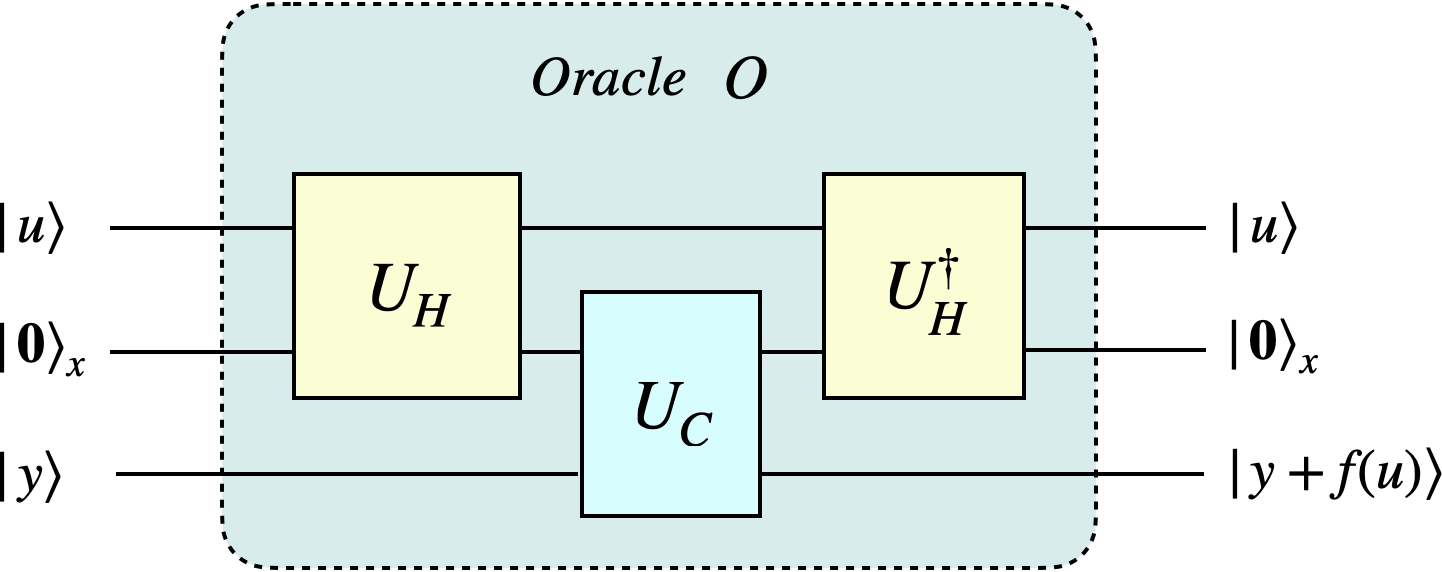}
\caption{For a given CSP, the internal structure of the quantum oracle can be constructed as: $O=U_{H}^\dag U_CU_{H}$, where $U_{H}=e^{-i\sum_{j=1}^k  \hat{H}_j\otimes \hat{p}_j}$, $U_C = P_C \otimes e^{-i\hat{p}} + (I - P_C) \otimes I$, and $P_C = \int_{\bm{z}\in C} \ket{\bm{z}}\bra{\bm{z}}d\bm{z}$. The first mode encodes the input state, and the second and the third modes are auxiliary modes. The states $\ket{\bm 0}_x$ and $\ket{y}$ correspond to the eigenstates of the position operators associated with the eigenvalues $\bm 0$ and $y$ respectively.
}\label{fig-oracle}
\end{figure}


\subsection{Circuit Implementation}\label{sec-impl}

In order to implement our algorithm on a CV quantum device, we need to figure out how to construct the Grover's iteration $G(\alpha,\beta)=R_\varphi(\alpha)R_t(\beta)$ using basic building blocks in CV quantum computers. Unlike $R_\varphi(\alpha)$, the physical construction of $R_t(\beta)$ is not straightforward due to the unknown nature of $\ket{t}$ prior to the search. However, we can construct $\tilde{R}_t(\beta)\equiv Z(-\beta)OZ(\beta)$, which is equal to $R_t(\beta)$ when restricted to the 2D subspace $\mathcal{T}$ as follows:
\begin{align}
\begin{split}
\tilde{R}_t(\beta)\ket{\bm{x}}\ket{0}_p &\equiv Z(-\beta)OZ(\beta)\ket{\bm{x}}\ket{0}_p \\
&=\begin{cases}
e^{i\beta}\ket{\bm{x}}\ket{0}_p, \ &\bm{x}\in Q \\
\ket{\bm{x}}\ket{0}_p, &\bm{x}\notin Q
\end{cases},
\end{split}
\end{align}
where $Z(\alpha) = e^{-i\alpha \hat{x}}$, $X(\beta) = e^{-i\beta \hat{p}}$, and $\ket{0}_p$ denotes the eigenstate of the momentum operator $\hat{p}$ with $p = 0$. Thus, when restricted to the 2D subspace $\mathcal{T}$, $\tilde{R}_t(\beta)$ is the same as $R_t(\beta)$. It is worthwhile to note that, during each iteration $G(\alpha,\beta)$ in Fig.~\ref{fig-circ-search}, only a single query to $O$ is required, which distinguishes it from the previous fixed-point quantum search algorithms where two queries per iteration are required~\cite{yoder2014fixed}.

Next, we discuss how to construct the quantum oracle $O$ for a given CSP. In query-based search problems, an oracle $f(\bm{x}): \bm{x} \in A\subset \mathbb{R}^d \to \{0,1\}$ is an indicator function to determine whether the input item $\bm{x}$ belongs to the target space $Q$. Without loss of generality, if $\bm{x} \in Q$, then $f(\bm{x}) = 1$; otherwise, $f(\bm{x}) = 0$. Currently, there exists an abstract form of the quantum oracle based on this indicator function $f(\bm{x})$, lacking a specific structure.

In this study, we investigate the internal structure of quantum oracles and propose a general framework applicable to solving optimization problems and spectral problems, as depicted in Fig.~\ref{fig-oracle}. For search problems, determining whether an input $\bm{x}$ belongs to $Q$ relies on analyzing a $k$-dimensional feature function $g(\bm{x})$. Such $g(\bm{x})$ encapsulates the $k$-dimensional features of $x$, and is defined and determined by the given CSP. Then we can express the classical oracle function $f(\bm{x})$ based on the feature function $g(\bm{x})$:
\begin{align}
f(\bm{x})=
\begin{cases}
1, & \text{if } g(\bm{x}) \in C\subset \R^k \\
0, & \text{if } g(\bm{x}) \notin C
\end{cases},
\end{align}
where the search criterion is to check whether $g(\bm{x})\in C$ and the target space is $Q = \{\bm{x} \mid g(\bm{x}) \in C\subset \R^k\}$. Note that the dimension $k$ of $g(\bm{x})$ is not necessarily the same as that of $\bm{x}$, since $g(\bm{x})$ only represents certain features of $\bm{x}$, not $\bm{x}$ itself. For example, if one wants to search for the tallest girl in Class B, then for a candidate such as \emph{Alice}, her feature function $g=[1.5,1]$ is two-dimensional, where the first entry represents her \emph{height} in meters and the second indicates her \emph{gender} (where $0$ represents male and $1$ for female). However, if one wants to search for the tallest student in Class B, then Alice's feature function simplifies to $g=1.5$, which is one-dimensional.

Consequently, the quantum oracle $O$ can be constructed directly from the feature function $g(\bm{x})$, without using the information from $f(\bm{x})$. Specifically, assuming $g(\bm{x})$ is $k$-dimensional, $O$ can be constructed as:
\begin{align}
O=U_{H}^\dag U_CU_{H},
\end{align}
where
\begin{align*}
U_{H}&=e^{-i\sum_{j=1}^k  \hat{H}_j\otimes \hat{p}_j},\\
U_C&= P_C\otimes e^{-i\hat{p}}+ (I-P_C) \otimes I,\quad P_C = \int_{\bm{z}\in C} \ket{\bm{z}}\bra{\bm{z}}d\bm{z},
\end{align*}
and $\{\hat{H}_j\}_{j=1}^k$ is determined by the given CSP. When $g(\bm{x})$ is one-dimensional, $U_{H}$ is reduced to $U_{H}=e^{-i \hat{H}\otimes \hat{p}}$. Note that the decision criterion for determining whether $g(\bm{x})\in C$ can be tailored to suit different CSPs, making this oracle construction flexible and adaptable to a wide range of problems~\cite{pati2000quantum,rebentrost2018photonic}.

Now we discuss how to choose $\{\hat{H}_j\}_{j=1}^k$ in $U_H$ and the input state $\ket{u}$ for the given CSP. In the application of finding the spectral of an operator $\hat{B}=\poly(\hat x,\hat p)$, where $\hat{B}$ is a polynomial in $\hat x$ and $\hat p$ of a single mode, the corresponding feature function $g$ is one-dimensional, and we can choose $\hat{H}\equiv \hat{B}$ and $U_{H}=e^{-i\hat{B}\otimes \hat{p}}$ in the oracle construction in Fig.~\ref{fig-oracle}. For this application, $\ket{u}$ denotes one eigenstate of $\hat{B}$ and the initial state can be written as a superposition of all eigenstates of $\hat{B}$, with details in Sec.~\ref{sec-eigen}. In fact, $U_H=e^{-i\hat{B}\otimes \hat{p}}$ in Fig.~\ref{fig-oracle} performs the continuous-version phase estimation algorithm with CV, where $g(\ket u)$ represents the corresponding eigenvalue of $\ket{u}$. Additionally, the second mode in Fig.~\ref{fig-oracle} serves as an auxiliary mode for implementing oracle operations and can be reused. In general, if the input state $\ket{u}=\ket{\bm{x}}$ in Fig.~\ref{fig-oracle}, then the effect of $O$ is described as follows:
\begin{align}
O\ket{\bm{x}}\ket{\bm{0}}_x\ket{y}=\ket{\bm{x}}\ket{\bm{0}}_x\ket{y+f(\bm{x})}.
\end{align}
After applying the oracle $O$, the feature information $g(\bm{x})$ will be extracted into the second register. For example, in the optimization problem in Sec.~\ref{sec-opt}, the role of the oracle is to extract the gradient information into the second register.

To better illustrate the effect of $g(\bm{x})$, we examine a CSP where $g(\bm{x})$ is one-dimensional. In this scenario, the indicator function $f(\bm{x})$ associated with the search problem acts as a binary classifier (e.g., a support vector machine or a neural network). Given an input item $\bm{x}$, its corresponding output value $g(\bm{x})$ is obtained by training a classifier for a binary classification problem~\cite{PhysRevResearch.1.033063}. The classification is then carried out based on the predefined criteria (e.g., whether $g(\bm{x}) \ge 0.5$ or not).

To summarize, the adjustable oracle construction $O=U_{H}^\dag U_CU_{H}$ developed here is crucially important for designing quantum search circuits to solve continuous optimization and spectral problems. Notably, our search algorithm can be generalized into an amplitude amplification method, enabling us to prepare quantum states with specific properties. For these applications, the key part is the design of the oracle $O$, which will be illustrated in detail in Sec.~\ref{sec-opt} and Sec.~\ref{sec-eigen}.

\section{Lower bound of query complexity for solving CSPs}

The optimality of the quantum search algorithm with a quadratic speedup for discrete search problems can be proven using quantum oracles~\cite{bennett1997strengths,boyer1998tight,zalka1999grover, Grover:2005:PQS:1073970.1073997,dohotaru2008exact}. However, for CSPs, it remains unclear what is the lower-bound of the query complexity for an arbitrary quantum search. Specifically, given a CSP with search space $A$, solution space $Q$, and the overlap $\lambda=\lambda_0=\frac{M(Q)}{M(A)}$. We say a quantum search algorithm solves the CSP with probability $p$, if after $q$ oracle queries, the overlap between the final state and the target state is at least $p$. Then we have the following theorem:

\begin{theorem}[Lower bound for solving CSPs]\label{t1}

If an arbitrary quantum search algorithm solves a given CSP with success probability $p$ after $q$ queries, then:
\begin{align}
q\ge \frac{1}{2\sqrt2}[(1+\sqrt{p} -\sqrt{1-p})\sqrt{n}-2],
\end{align}
where $n=\lceil \frac{1}{\lambda} \rceil$.
\end{theorem}

Theorem~\ref{t1} states that, for CSPs, the lower bound of query complexity among all query-based quantum search algorithms is $\mathcal{O}(1/\sqrt{\lambda})$, implying that our quantum search algorithm for CSPs is optimal in terms of query complexity.

\section{Application to optimization problems}\label{sec-opt}

\begin{figure}
\centering
\includegraphics[width=1\columnwidth]{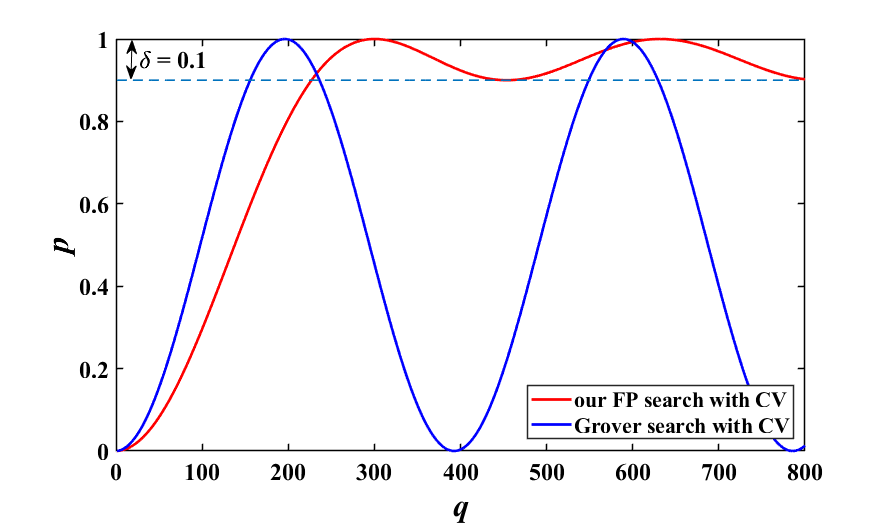}
\caption{The simulation of the quantum search algorithm to solve the optimization problem. $q$ is the number of iterations and $p$ is the probability that the searched quantum state is projected to the target state. As illustrated, the probability $p$ first increases with the number of iterations, and then oscillates in the interval where the accuracy is greater than 0.9 ($\delta = 0.1$).}\label{fig-simulation}
\end{figure}

\begin{table*}[htbp]
\centering
\setlength{\tabcolsep}{1pt}
\begin{adjustbox}{width=\textwidth}
\begin{tabular}{>{\centering\arraybackslash}p{3cm}
                @{\hspace{0.01cm}}
                >{\centering\arraybackslash}p{6cm}
                @{\hspace{0.01cm}}
                >{\centering\arraybackslash}p{4cm}
                @{\hspace{0.01cm}}
                >{\centering\arraybackslash}p{3.5cm}
                @{\hspace{0.01cm}}
                >{\centering\arraybackslash}p{3.5cm}}
\hline
\textbf{Function Name} & \textbf{Function Form} & \textbf{Search Space} & \textbf{Quantum} \newline \textbf{Required Iterations} & \textbf{Classical} \newline \textbf{Expected Iterations} \\
\hline
rastrigin~\cite{rastrigin1974systems} & $\sum_{i=1}^{2} \left(x_i^2 - 10 \cos(2 \pi x_i)\right)$ & $-2 \le x_i \le 2$ & 353 & $1.3872 \times 10^5$ \\
\hline
styblinski tang~\cite{styblinski1990experiments} & $\frac{1}{2} \sum_{i=1}^{2} \left(x_i^4 - 16x_i^2 + 5x_i\right)$ & $-2 \le x_i \le 2$ & 147 & $2.406 \times 10^4$ \\
\hline
alpine02~\cite{Al-Roomi2015} & $-\prod_{i=1}^{2} \left(\sqrt{x_i} \sin(x_i)\right)$ & $0 \le x_i \le 10$ & 15 & 237 \\
\hline
Himmelblan~\cite{himmelblau1972applied} & $(x_1 + x_2 - 11)^2 + (x_1 + x_2^2 - 7)^2$ & $-2 \le x_i \le 2$ & 256 & $7.272 \times 10^4$ \\
\hline
Rosenbrock~\cite{10.1093/comjnl/3.3.175} & $(1 - x_1)^2 + 100(x_2 - x_1^2)^2$ \newline s.t. \ $x_1^2 + x_2^2 \le 2$ & $x_1^2 + x_2^2 \le 2$ & 237 & $6.262 \times 10^4$ \\
\hline
Gomez and Levy~\cite{levy1985tunneling} & $4x_1^2 - 2.1x_1^4 + \frac{1}{3}x_1^6 + x_1x_2 - 4x_2^2 + 4x_2^4$ \newline s.t. \  $-\sin(4\pi x_1)+2\sin^2(2\pi x_2)\le 1.5$ & $-1 \le x_1 \le 0.75$ \newline $-1\le x_2\le 1$ & 58 & 3786 \\
\hline
\end{tabular}
\end{adjustbox}

\caption{Simulation results of applying our algorithm to a group of test functions in optimization. We use the fixed-point quantum search method with CV to find stationary points, where the partial derivatives do not exceed 0.1. Given a projection probability greater than 0.9, the results of our quantum algorithm show significant advantages over classical search methods.}
\label{table1}
\end{table*}

One crucial application of quantum search is that it can be applied to solve optimization problems~\cite{durr1999quantumalgorithmfindingminimum,baritompa2005grover,conti_opti_morimoto}. Compared to established quantum search algorithms based on DV models, one benefit of our search algorithm for continuous optimization problems lies in that there is no need for discretization, and no need to increase the number of qubits to gain better discretization accuracy. Hence, our proposed algorithm offers a robust solution to a wide range of problems in constrained optimization~\cite{constrained_opti}, nonconvex optimization~\cite{nonconvex-opti}, derivative-free optimization~\cite{deriv_free_opti}, and energy landscape~\cite{energy_lands}. Another benefit is that our algorithm is able to identify all optimal and near-optimal solutions within a given search space; in comparison, gradient-based algorithms, such as the gradient descent or the Newton method only converge to one local optimal solution from an initial trial solution.

In the following, we focus on the constrained optimization problems to demonstrate our algorithm. Specifically, a general constrained optimization problem is formulated as follows: we aim to find the optimal values of an objective function $h(\bm{x}): \mathbb{R}^d \to \mathbb{R}$ while satisfying a set of constraints. Without loss of generality, we denote the region that satisfies all the constraints as $A\subset \R^d$. Then all optimal and near-optimal solutions $\bm x \in A$ must satisfy $\left\|\nabla h(\bm{x})\right\|_\infty \le \epsilon$, for a given small $\epsilon$. Hence, we can construct the following oracle function $f$ to characterize the optimization solutions:
\begin{align*}
f(\bm{x}) =
\begin{cases}
1, & \left\|\nabla h(\bm{x})\right\|_\infty \le \epsilon\text{ and } x\in A \\
0, & \left\|\nabla h(\bm{x})\right\|_\infty > \epsilon\text{ or } x\notin A
\end{cases},
\end{align*}
where $f(\bm{x})$ determines whether the gradient magnitude at $\bm{x}$ is within the acceptable range defined by $\epsilon$.

Next, we follow our algorithm to prepare the initial state $\ket\varphi$ and define the target space $Q$ and the target solution $\ket t$. Then we construct the quantum search circuit shown in Fig.~\ref{fig-circ-search}. Below, we discuss how to design the oracle $O$ in Fig.~\ref{fig-oracle} for this optimization problem. We rewrite the unitary $e^{-i\hat{H}\otimes \hat{p}}$ as $e^{-i\sum_{j=1}^d \partial_j h(\bm{\hat{x})}\otimes \hat{p}_{d+j}}$, since we need information about each component of the gradient. Then the criterion for $U_C$ can be modified to require that the partial derivative of each component $|\partial_j h(x)|$ is less than $\epsilon$.

Assuming that the input state of the first mode is $\ket{\varphi} = \int_A \varphi(\bm{x}) \ket{\bm{x}} d\bm{x}$, the evolution is as follows:
\begin{align}
\begin{split}
O \ket{\varphi} \ket{\bm{0}}_x \ket{y}
=& \int_{Q} \varphi(\bm{x}) \ket{\bm{x}} d\bm{x} \otimes \ket{\bm{0}}_x \otimes \ket{y+1} \\
&+ \int_{A/Q} \varphi(\bm{x}) \ket{\bm{x}} d\bm{x} \otimes \ket{\bm{0}}_x \otimes \ket{y},
\end{split}
\end{align}
where the second register is $d$ auxiliary modes $\ket{0}_x^{\otimes d}$ and can be reused. Ignoring the second register, it can be seen that the oracle based on the indicator function $f(\bm{x})$ is implemented as $O \ket{\bm{x}} \ket{y} = \ket{\bm{x}} \ket{y + f(\bm{x})}$.

To demonstrate the effectiveness of our algorithm, we apply it to minimize the Rosenbrock function, which is a constrained non-convex optimization problem~\cite{10.1093/comjnl/3.3.175}:
\begin{align}
\begin{split}
\min\ &h(x_1,x_2)=(1-x_1)^2 + 100(x_2-x_1^2)^2,\\
s.t.  \ &x_1^2+x_2^2\le 2.
\end{split}
\end{align}
Here, the search space is $A=\{(x_1,x_2)|x_1^2+x_2^2\le 2\}$, and we aim to search for the target space $Q=\{(x_1,x_2)| |\partial_j h(x_1,x_2)| \le 0.1\}$ in the search space. In this search problem, the initial state is $\ket{\varphi} = U\ket{s} = \frac{1}{\sqrt{2\pi}}\int_A \ket{\bm{x}} \, d\bm{x}$, and the target state is $\ket{t} = \frac{1}{\sqrt{|Q|}} \int_Q \ket{\bm{x}} \, d\bm{x}$. We convert this problem into a 2D subspace spanned by $\ket{t}$ and $\ket{\bar{t}}$, where $\ket{\bar{t}}$ represents a non-target solution. The simulation results in Figure~\ref{fig-simulation} demonstrate that the red line (our fixed-point quantum search with CV) oscillates within the required accuracy range $1-\delta$ as the number of iterations increases beyond a certain value, where we choose $\delta=0.1$. In contrast, the blue line (the initial version of the quantum search) exhibits significant oscillations, resulting in both ``overcooking" and ``undercooking". This indicates that our proposed fixed-point quantum search algorithm with CV can effectively solve CSPs.

Besides the Rosenbrock function, we also apply our algorithm to solve five other typical test functions in optimization, and some of these functions are challenging to optimize due to the presence of a large number of local minima. The simulation results are summarized in Table~\ref{table1}. For each of the six test functions, we apply our fixed-point quantum continuous search algorithm to look for stationary points where the gradient magnitude becomes vanishingly small. Specifically, the search criterion is set such that the norm of the gradient is no greater than $0.1$ with a required success probability of $p>0.9$. By simulating the quantum search algorithm on a classical CPU server, we find the minimum number of iterations to ensure the success probability exceeds $0.9$ for each optimization problem. For comparison, Table~\ref{table1} records both the number of iterations required by our quantum algorithm and the iterations needed for a direct classical search. The results clearly demonstrate that the quantum search requires significantly fewer iterations than the classical counterpart for each optimization problem, showcasing a distinct quadratic advantage.

It is worthwhile to note that like all other quantum search algorithms, our quantum search algorithm for CSPs works well in theory. In practice, the success of our algorithm relies on the successful implementations of quantum error correction on noisy quantum devices~\cite{Cai2021,Joshi_2021}. Without error correction, the circuit of our algorithm can be very sensitive to noise. In order to visualize how noise could accumulate and damage the quantum coherence of our algorithm, we add depolarized noise to the quantum circuit in Fig.~\ref{fig-circ-search}. When restricting to the two-dimensional subspace spanned by $\{\ket{t},\ket{\bar t}\}$, the noise dynamics of the system state $\rho$ after each iteration is approximately governed by:
\begin{align}
\label{eqn:depolar}
\epsilon(\rho)= (1-\eta)\rho+ \frac{\eta}{2} I,
\end{align}
where the depolarizing probability $\eta \in[0,1]$ characterizes the noise intensity. In Fig.~\ref{fig-simulation-noise}, we perform the simulation under the dynamics of Eq.~(\ref{eqn:depolar}) to optimize the `alpine02' function in Table~\ref{table1} under various noise intensities $\eta$ ranging from 0 to 0.03. At depolarizing noise levels of $\eta=0.005$ or lower, the algorithm still finds the target state with higher than $90\%$ probability after $17$ iterations, demonstrating reasonable noise resilience. When $\eta$ gets larger than 0.02, the final output state of our algorithm will significantly deviate from the target state. So while low noise levels are tolerated, success probability degrades significantly as depolarizing noise increases, highlighting the importance of error correction for real-world implementation.

\begin{figure}
\centering
\includegraphics[width=1\columnwidth]{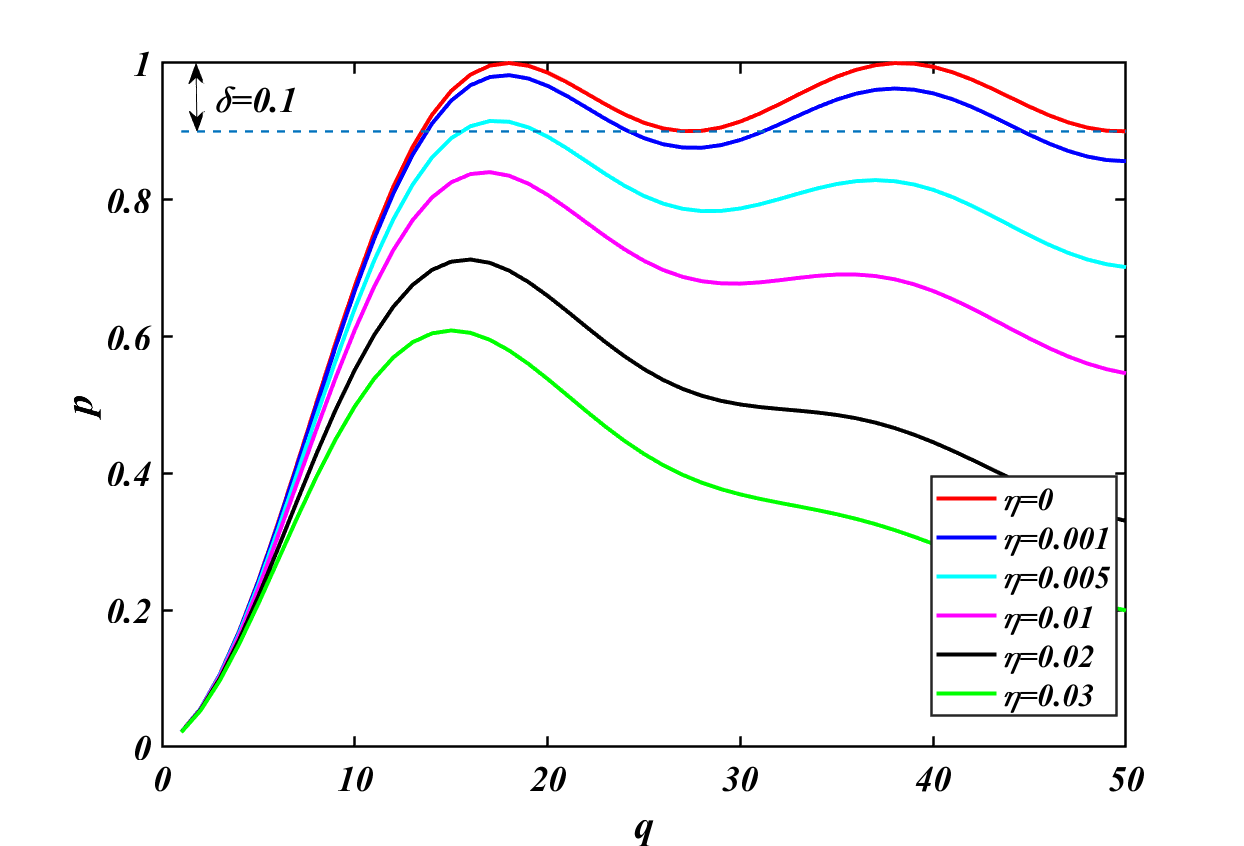}
\caption{Simulations of applying our algorithm to optimize the `alpine02' function under different magnitudes of noise intensity, where $q$ is the number of iterations in our algorithm and $p$ is the success probability. To better illustrate the results, we set an error threshold $\delta=0.1$ (dashed line). The noise effect is modeled by a quantum depolarizing channel, with the depolarizing probability $\eta$ to characterize the noise intensity, ranging from $0$ to $0.03$. For sufficient low value of $\eta$, the algorithm works reasonably well; as the noise gets larger, the systems's final state will deviate more from the target state.
}\label{fig-simulation-noise}
\end{figure}

\section{Application to spectral problems}\label{sec-eigen}

One interesting application of our quantum search algorithm is that we can apply it to solve spectral (or eigenvalue) problems for operators defined in CV quantum information. Specifically, given an analytic function $F(x,p)$, based on $\hat x$ and $\hat p$, we can define a self-adjoint operator $\hat B=F(\hat x,\hat p)$ based on functional calculus~\cite{quant_math_hall}. Typical examples of such $\hat{B}$ include the polynomial operators $\hat x^3$, $\hat p^3$, $\hat x+\hat p$, $\hat x^2+\hat p^2$, $\hat x^2-\hat p^2$, and $\hat x^4+\hat p^2$. The spectrum of $\hat B$ can be discrete, continuous, or a mixture of both. Unfortunately, except for a few examples, an analytic way of finding out the spectrum of $\hat B$ is extremely difficult. However, we can define the following eigenvalue problem: for a given self-adjoint operator $\hat{B}$, we aim to find the eigenstates of $\hat B$ whose corresponding eigenvalues fall into the interval $[a,b]$. Such eigenvalue problems can be conveniently solved by our quantum search algorithm on a CV quantum device.

In the following, we solve this eigenvalue problem for a self-adjoint operator $\hat B=\poly(\hat x,\hat p)$, where $\poly(x,p)$ is a polynomial in $x$ and $p$. In the quantum circuit of Fig.~\ref{fig-oracle}, we set $\hat{H} = \hat{B}$ and define $U_C$ as follows:
\begin{align*}
U_C=U_{[a,b]}=P_{[a,b]} \otimes e^{-i\hat p} + (I - P_{[a,b]}) \otimes I,
\end{align*}
where the projection operator $P_{[a,b]}$ is given by $P_{[a,b]}= \int_{a}^b \ket{x} \bra{x} \, dx$. For the DV-model case, the quantum phase estimation can be used to construct the quantum oracle to solve the eigenvalue problem~\cite{jin2020query}; similarly, for the continuous-spectrum case, we can construct the quantum oracle $O$ based on $e^{-i\hat{B}\otimes \hat{p}}$, which is the quantum phase estimation circuit for CV systems. For $\hat B=\poly(\hat x,\hat p)$, the unitary $e^{-i\hat{B}\otimes \hat{p}}$ can be physically constructed using the method in~\cite{lloyd1999quantum,kalajdzievski2019exact,PhysRevResearch.6.023332}. For example, if $\hat{B}=\theta_1(\hat{p}+3\theta_2\hat{x}^2)$, then we can design the following CV quantum circuit to generate $e^{-i \hat{B}\otimes\hat{p}}$:
\begin{align}
\begin{split}
e^{-i \hat{B}\otimes\hat{p}}&= e^{-i\theta_1(\hat{p}+3\theta_2\hat{x}^2)\otimes\hat{p}} \\
&= e^{-i\theta_2\hat{x}^3}e^{-i\theta_1\hat{p}\otimes\hat{p}}e^{i\theta_2\hat{x}^3}.
\end{split}
\end{align}

If $\hat{B}$ has a continuous spectrum, then the eigenvalue equation is given by $\hat{B}\ket{u_\alpha}=E_\alpha\ket{u_\alpha}$, where $\ket{u_\alpha}$ is the eigenstate corresponding to the eigenvalue $E_\alpha$. For the input state $\ket{\varphi}$, it can be expanded in terms of $\{\ket{u_\alpha}\}$, leading to $\ket{\varphi}=\int \varphi_\alpha \ket{u_\alpha} d\alpha$. The evolution under $e^{-i\hat{B}\otimes \hat{p}}$ is governed by:
\begin{align}
e^{-i \hat{B} \otimes \hat{p}} \ket{\varphi} \ket{0}_x = \int \varphi_\alpha \ket{u_\alpha} \ket{E_\alpha}_x d\alpha.
\end{align}
In Fig.~\ref{fig-oracle}, the quantum oracle $O$ becomes:
\begin{align}
O=O_{[a,b]}=(e^{i \hat{B} \otimes \hat{p}}\otimes I)(I\otimes U_{[a,b]})(e^{-i \hat{B} \otimes \hat{P}}\otimes I).
\end{align}
Then the evolution under $O$ is given by:
\begin{align}
\begin{split}
O \ket{\varphi} \ket{0}_x \ket{y} &= \int_{\alpha:E_\alpha\in [a,b]} \varphi_\alpha \ket{u_\alpha} \ket{0}_x \ket{y+1}d\alpha \\
&+ \int_{\alpha:E_\alpha\notin [a,b]} \varphi_\alpha \ket{u_\alpha} \ket{0}_x \ket{y}d\alpha.
\end{split}
\end{align}
Ignoring the second mode, the above formula shows that $O=O_{[a,b]}$ is indeed the oracle we need: $O_{[a,b]} \ket{u_\alpha} \ket{y} = \ket{u_\alpha} \ket{y + f(\ket{u_\alpha})}$, where $f$ is an indicator function used to determine whether the eigenvalue $E_\alpha$ corresponding to the eigenstate $\ket{u_\alpha}$ lies within $[a,b]$. Hence, by applying our fixed-point quantum search algorithm for CSPs, we can obtain a final state very close to $\ket t$, where
\begin{align}
\ket{t} = \frac{1}{Z}\int_{\alpha:E_\alpha\in [a,b]} \varphi_\alpha \ket{u_\alpha} d\alpha
\end{align}
is a superposition of all the required eigenstates, and $Z$ is the normalization factor.

If $\hat{B}$ has a discrete spectrum, then the eigenvalue equation reads $\hat{B}\ket{u_i}=E_i\ket{u_i}$. For the input state $\ket{\varphi}$, it can likewise be expanded according to the eigenstates, leading to $\ket{\varphi}=\sum_j \varphi_j\ket{u_j}$, and our quantum algorithm can be applied in a similar way.

\section{Conclusion and discussion}

Our algorithm for CV quantum system can be applied to a wide range of bosonic quantum platforms, including optical modes~\cite{15db,Su2013}, phononic modes in trapped-ion system~\cite{Chen2023}, mechanical oscillators~\cite{VonLupke2024}, and superconducting circuits~\cite{Cai2021}. To implement our algorithm on a CV quantum computing platform, two key challenges must be addressed: (1) the capability to implement non-Gaussian gates, such as the cubic phase gate, which are essential for universal computation, and (2) the decoherence, such as the population loss in superconducting bosonic modes. Despite these hurdles, substantial experimental progress has been achieved. High-fidelity Gaussian operations, including squeezing, beam splitting, and displacement, are now routinely implemented in modern optical laboratories. Notable achievements include generating 15 dB of optical squeezing~\cite{15db}, producing large-scale two-dimensional CV cluster states~\cite{2d-cluster}, and generating CV gate sequence consisting of a single-mode squeezing gate and a two-mode controlled-phase gate~\cite{Su2013}. While significant challenges remain, the field continues to advance rapidly, driven by innovative experimental techniques. Furthermore, emerging proposals for CV quantum algorithms are poised to provide additional momentum, accelerating the path toward practical and scalable CV quantum computing.

In summary, for CSPs, we have designed a fixed-point quantum search algorithm using the CV model that drives the system state to converge to the solution target state with a quadratic quantum speedup. We have also proved that our quantum search algorithm is optimal in achieving the lower-bound of the query complexity for an arbitrary quantum search algorithm. Compared to discrete search algorithms, our approach circumvents issues related to increased resource consumption and loss of accuracy caused by discretization. Moreover, this algorithm is well-suited for amplitude amplification of target states in the CV model. We have also provided a concrete circuit implementation of our algorithm and discussed how to design the internal structure of the quantum oracle for a given CSP. This framework allows the algorithm to be applied effectively to solve both optimization and eigenvalue problems as applications.

\begin{acknowledgements}
We gratefully acknowledge the supports from the National Natural Science Foundation of China under Grant No. 92265208, 11925404, 92165209, 92365301, and 92265210, and Innovation Program for Quantum Science and Technology under Grant No.~2021ZD0300200.
\end{acknowledgements}

\bibliography{ref}

\clearpage
\appendix

\section{$\pi/3$ search algorithm}\label{sec-pi3-search}

In Section~\ref{sec:search_design}, we have proposed a fixed-point quantum search algorithm based on continuous variables (CV) that provides a quadratic speedup $\O(1/\sqrt{\lambda})$ over classical counterparts. Here, for comparison, we develop a CV adaptation of the $\pi/3$ fixed-point quantum search algorithm, initially proposed by Grover in 2005~\cite{grover2005fixed}. In essence, it replaces the phase inversion operation with the $\pi/3$ phase rotation operator and ensures the quantum state asymptotically converges to the target state $\ket{t}$. However, the drawback of such algorithm is that it has a query complexity of $\O(1/\lambda)$, i.e., no quantum advantage.

Specifically, we construct the initial state $\ket{\varphi}$ and define the target state $\ket{t}$:
\begin{align*}
\ket\varphi &= \int_A \varphi (\bm x) \ket{\bm{x}} d\bm x,\\
\ket t &\equiv \int_{A} \varphi_1 (\bm{x}) \ket{\bm x}  d\bm{x}=\frac{1}{\sqrt{I_1}}\int_{Q} \varphi (\bm{x}) \ket{\bm x}  d\bm{x},
\end{align*}
where $I_1=\int_{Q} |\varphi (\bm{x})|^2 d\bm{x}$ and $\varphi_1 (\bm{x})=\varphi (\bm{x})/\sqrt{I_1}$ on $Q$ and zero on $A/Q$. We assume $\ket{\varphi}$ is generated from the source state $\ket{s}$ through the unitary $U$ satisfying $\ket{\varphi} = U \ket{s}$.

The Grover's iteration sequence is characterized by following two unitary transformations
\begin{align*}
R_t &= I - \left(1 - e^{\frac{\pi}{3}i}\right) \ket{t}\bra{t},\\
R_s &= I - \left(1 - e^{\frac{\pi}{3}i}\right) U^\dagger \ket{\varphi} \bra{\varphi} U.
\end{align*}
$R_s$ and $R_t$ represent the selective phase shift operations for states $\ket{s}$ and $\ket{t}$, respectively, with phase $\pi/3$, and can be shown to be unitary. To avoid ``undercooking" and ``overcooking," we use a recursive method to construct the unitary transformation $U_{m}$ to construct the fixed-point quantum search:
\begin{align}\label{eq-pi3-u}
{{U}_{m}} = {{U}_{m-1}}{{R}_{s}}U_{m-1}^{\dagger}{{R}_{t}}{{U}_{m-1}}, \qquad U_0 = U,
\end{align}
where $m$ denotes the number of recursion steps, satisfying $3^m = 2q + 1$. In this work, query complexity is defined as the number of oracle calls needed to meet a predefined accuracy for the target state. We can evaluate the performance of this $\pi/3$ fixed-point quantum search algorithm by analyzing its query complexity.

The overlap between the target state $\ket{t}$ and the initial state $\ket{\varphi}$ is defined as $\lambda \equiv |\bra{t}U\ket{s}|^2$. For convenience, we introduce a parameter $\epsilon \in [0,1)$ to describe the overlap $\lambda = 1 - \epsilon$ since $0 < \lambda \le 1$. According to Eqn. \eqref{eq-pi3-u}, we can prove by induction that:
\begin{align}\label{eq-pi3-epsilon}
{{| \left\langle  t \right|{{U}_{m}}\left| s \right\rangle  |}^{2}} = 1 - {{\varepsilon }^{{{3}^{m}}}} = 1 - {{\varepsilon }^{2q + 1}},
\end{align}
where $q$ is the number of queries and $0 \le \epsilon < 1$. From this, it follows that as the number of queries increases, the error between the final state $U_m \ket{s}$ and the target state $\ket{t}$ decreases to zero, ensuring that the quantum state asymptotically converges to $\ket{t}$. In order to satisfy $\left|\bra{t}U_m\ket{s}\right|^2 \ge 1 - \delta$, where $\delta$ is the desired error tolerance, the required number of queries can be expressed as $q = \left(\frac{\ln \delta}{\ln(1 - \lambda)} - 1\right)/2$, derived from Eqn. \eqref{eq-pi3-epsilon}. When the overlap $\lambda$ between the initial state $\ket{\varphi}$ and the target state $\ket{t}$ is small enough, the number of queries required can be rewritten as:
\begin{align}
q_{\frac{\pi}{3}} = \frac{1}{2} \left(-\frac{\ln \delta}{\lambda} - 1\right) \sim \mathcal{O}\left(\frac{1}{\lambda}\right).
\end{align}
Thus, this $\pi/3$ fixed-point search algorithm for CSPs shows no quantum advantage over classical algorithms in terms of query complexity.

\section{Lower bound for solving CSPs}\label{sec-lower-bd}

In Section~\ref{sec:search_design}, for CSPs, we propose a CV fixed-point quantum search algorithm that provides a quadratic speedup $\O(1/\sqrt{\lambda})$ over classical algorithms, where $\lambda=\lambda_0=\frac{m(Q)}{m{(A)}}$ denotes the overlap between the target solution space and the search space. Here, we establish the lower bound on query complexity of an arbitrary query-based quantum algorithm for solving CSPs.

Previous studies~\cite{bennett1997strengths, boyer1998tight, dohotaru2008exact} have rigorously proved the optimality of the Grover's search algorithm in achieving the quadratic speedup over classical algorithms. These studies show that any quantum algorithm to solve the query-based search problem must have a query complexity no smaller than $\mathcal{O}(\sqrt{\frac{N}{M}})$, where $N$ is the number of initial items and $M$ is the number of target items. Similarly, we can derive a lower bound on query complexity for any quantum search in the CV framework.

In the following, we use the state-vector norm derived from the inner product, with $\|\ket \psi  \|^2 = |\langle \psi|\psi \rangle|$. The quantum states used in this work are all normalized states, that is, any quantum state $\psi$ satisfies $\|\ket \psi  \|=1$. It is straightforward to see that this norm satisfies unitary invariance and non-negativity, i.e., $\left \| U|a\rangle \right \|=\left \| |a\rangle \right \| \ge 0$.

\begin{lemma} For $0\le a_i \le 1$ and $\sum_{i=1}^{n} a_i^2=1$, we have:
\begin{align}
\max \sum_{i=1}^{n} a_i = \sqrt n.
\end{align}\label{lemma_1}
\end{lemma}

\begin{proof}
Obtained from the Cauchy-Schwarz inequality,
\begin{align*}
\begin{split}
(\frac{1}{n}\sum_{i=1}^{n} a_i)^2
&= (\frac{a_1\cdot 1+a_2\cdot 1+...+a_n\cdot 1}{n})^2\\
&\le (\frac{a_1^2+a_2^2+...+a_n^2}{n^2})(1+1+...+1)\\
&= \frac{a_1^2+a_2^2+..+a_n^2}{n}=\frac{1}{n}.
\end{split}
\end{align*}
Hence,
\begin{align*}
\sum_{i=1}^{n} a_i \le \sqrt n,
\end{align*}
where equality holds if and only if all $a_i$ are equal.
\end{proof}

Without loss of generality, for a CSP with search space $A$ and solution space $Q$, an arbitrary quantum search algorithm to solve this CSP can be characterized by the following unitary sequence and the final state $\ket{\psi_Q^{q}}$:
\begin{align}\label{eqn:search_u_sequence}
\ket{\psi_Q^{q}} \equiv U_q O_{Q}(\alpha_q) U_{q-1} \cdots U_1 O_Q(\alpha_1) U_0 \ket{s},
\end{align}
where $\ket{s}$ is the source state and $U_i$ is an arbitrary unitary transformation. In the gate sequence, the oracles $O_Q(\alpha_j)$, $j=1,\ldots,q$ are determined by the target solution space $Q$ and the rotational phase parameter $\alpha_i$. Unlike the conventional quantum oracles which flip the target state by phase $\pi$, the oracle $O_Q(\alpha)$ describes flipping the target state by an arbitrary phase $\alpha$:
\begin{align}
O_Q(\alpha) \ket{x} =
\begin{cases}
e^{i\alpha} \ket{x}, & x \in Q \\
\ket{x}, & x \notin Q
\end{cases}.
\end{align}
For example, when the search algorithm is the original Grover's algorithm, $O_Q(\alpha)$ has a flip phase $\alpha=\pi$. The overlap between the initial state $U \ket{s}$ and the target state $\ket{t}$ is $\lambda = |\langle t | U | s \rangle|^2$. We partition the entire search space $A$ into subregions $Q_1, Q_2, \ldots, Q_n$ with $n = \lceil \frac{1}{\lambda} \rceil$.  For each $Q_i$ of the partition, assuming it corresponds to the solution space of some CSP, then the quantum search algorithm expressed in Eqn.~(\ref{eqn:search_u_sequence}) will generate $\ket{\psi_{Q_i}^{q}}$.

Next, based on the quantum search sequence in Eqn.~(\ref{eqn:search_u_sequence}), we define
\begin{align}
\ket{\psi^q} \equiv U_q U_{q-1} \cdots U_1 U_0 \ket{s}.
\end{align}
The difference between $\ket{\psi_Q^q}$ and $\ket{\psi^q}$ is that $\ket{\psi^q}$ does not include $O_Q(\alpha_j)$, $j=1,\ldots,q$.

In addition, we define
\begin{align*}
\ket{\psi_Q^{j,q}} \equiv U_q U_{q-1} \cdots U_{j+1}\big(U_jO_Q(\alpha_j)\big) \cdots \big(U_1 O_Q(\alpha_1)\big) U_0 \ket{s}.
\end{align*}

For each subregion $Q_i$, we can calculate the distance $d_{q,i}=\left\| \ket{\psi^q} - \ket{\psi_{Q_i}^q}  \right\|$. Then we can define the average distance $ d_q\equiv  \frac{1}{n}\sum_{i=1}^{n}d_{q,i}=\frac{1}{n}\sum_{i=1}^{n}\left\| \ket{\psi^q} - \ket{\psi_{Q_i}^q}  \right\|$.

\begin{lemma}[Upper bound of distance by $q$ queries]
The maximum of the average distance $d_q$ after $q$ queries is $\frac{2q}{\sqrt n}$.\label{lemma_2}
\end{lemma}

\begin{proof}
We use the triangle inequality, unitary invariance and Lemma~\ref{lemma_1} to get,
\begin{align}
\begin{split}
d_q=&\frac{1}{n}\sum_{i=1}^{n}\left\| \ket{\psi^q} - \ket{\psi_{Q_i}^q}  \right\|\\
\le& \frac{1}{n}\sum_{i=1}^{n} \sum_{j=0}^{q-1} \left\| \ket{\psi_{Q_i}^{j+1,q}} - \ket{\psi_{Q_i}^{j,q}} \right \| \\
=& \frac{1}{n}\sum_{j=0}^{q-1} \sum_{i=1}^{n} \left\| O_{Q_i}(\alpha_{j+1}) \ket{\psi_{Q_i}^j}- \ket{\psi_{Q_i}^j} \right \|\\
\le& \frac{2}{n}\sum_{j=0}^{q-1} \sum_{i=1}^{n} \left\| P_{Q_i} \ket{\psi_{Q_i}^j} \right \| \\
\le& 2\sum_{j=0}^{q-1} \frac{1}{\sqrt n} = \frac{2q}{\sqrt n},
\end{split}
\end{align}
where the second inequality holds for all $\alpha_{j+1} = \pi$ and the last inequality uses the Lemma~\ref{lemma_1}.
\end{proof}


\begin{lemma}[Lower bound of distance by $q$ queries]
We assume that the quantum state $\ket{\psi_Q^q}$ after $q$ queries is projected onto the target state $\ket{t}$ with a probability of at least $p$. Then the minimum of the average distance $d_q$ after $q$ queries is $\frac{1}{\sqrt{2}} \left(1 + \sqrt{p} - \sqrt{1-p} - \frac{2}{\sqrt{n}}\right)$.\label{lemma_3}
\end{lemma}

\begin{proof}
We use the triangle inequality, the possibility at least $p$, and basic inequality to get,
\begin{align}
\begin{split}
d_{q,i}=&\left \| \ket{\psi^q}-\ket{\psi_{Q_i}^q} \right \|\\
\ge &\left \| P_{Q_i}(\ket{\psi_{Q_i}^q} - \ket{\psi^q}) \right \| - \left \| P_{Q_i}^\perp(\ket{\psi_{Q_i}^q} - \ket{\psi^q}) \right \|\\
\ge &\frac{1}{\sqrt2}(\left \| P_{Q_i}\ket{\psi_{Q_i}^q} \right \| - \left \| P_{Q_i}\ket{\psi^q} \right \|\\ &+\left \| P_{Q_i}^\perp\ket{\psi^q} \right \| - \left \| P_{Q_i}^\perp\ket{\psi_{Q_i}^q} \right \|)\\
\ge &\frac{1}{\sqrt2}(\sqrt p - \sqrt{1-p}+1-2\left \| P_{Q_i}\ket{\psi^q} \right \|),
\end{split}
\end{align}
where $P_{Q_i}^\perp$ is the projection operator of the complement of $Q_i$.

On the basis of the above inequality, the average distance after $q$ queries is,
\begin{align}
\begin{split}
d_q=&\frac{1}{n}\sum_{i=1}^{n}\left \| \ket{\psi^q}-\ket{\psi_{Q_i}^q} \right \|\\
\ge & \frac{1}{n}\sum_{i=1}^{n}\frac{1}{\sqrt2}(\sqrt p - \sqrt{1-p}+1-2\left \| P_{Q_i}\ket{\psi^q} \right \|)\\
\ge & \frac{1}{\sqrt2}(1+\sqrt{p}-\sqrt{1-p}-\frac{2}{\sqrt n}),
\end{split}
\end{align}
where the last inequality can be obtained by Lemma~\ref{lemma_1}.
\end{proof}

Combining Lemma~\ref{lemma_2} and Lemma~\ref{lemma_3}, we can derive the lower bound of $q$:
\begin{align}
q \ge \frac{1}{2\sqrt{2}} \left[(\sqrt{p} - \sqrt{1-p} + 1) \sqrt{n} - 2\right].
\end{align}
Since this lower-bound result holds for arbitrary quantum search sequence in Eqn.~(\ref{eqn:search_u_sequence}), the optimal query complexity among all query-based quantum search algorithms is $\mathcal{O}(1/\sqrt{\lambda})$. Thus, the result in Theorem~\ref{t1} is proved. This implies that our quantum search algorithm for CSPs achieves the optimal query complexity.

\end{document}